\pdfoutput=1
\documentclass[conference]{IEEEtran}
\IEEEoverridecommandlockouts

\usepackage{cite}
\usepackage{amsmath,amssymb,amsfonts}
\usepackage{amsthm}
\usepackage{graphicx}
\usepackage{booktabs}
\usepackage{multirow}
\usepackage{array}
\usepackage{xcolor}
\usepackage{listings}
\usepackage{tikz}
\usetikzlibrary{arrows.meta,positioning,fit,backgrounds,calc}
\usepackage{url}
\usepackage[hidelinks]{hyperref}

\newtheorem{proposition}{Proposition}

\lstdefinestyle{ebpf}{
  language=C,
  basicstyle=\ttfamily\scriptsize,
  keywordstyle=\color{blue!60!black}\bfseries,
  commentstyle=\color{green!40!black}\itshape,
  stringstyle=\color{red!60!black},
  showstringspaces=false,
  columns=fullflexible,
  keepspaces=true,
  numbers=left,
  numberstyle=\tiny\color{gray},
  numbersep=4pt,
  xleftmargin=10pt,
  frame=tb,
  breaklines=true,
  morekeywords={__u8,__u32,__u64,__be16,__be32,SEC,bool}
}

\begin{document}

\title{Line-Rate GTP-U Admission Control at the Edge of a Cloud-Native 5G Core: An XDP-Based Design for Kubernetes-Hosted User Plane Functions}

\author{\IEEEauthorblockN{Simhadri Podala Narasimha}
\IEEEauthorblockA{\textit{Independent Researcher}\\
Littleton, Colorado, USA\\
podalas@gmail.com}}

\maketitle

\begin{abstract}
The User Plane Function (UPF) of a 5G Standalone core terminates every GPRS Tunnelling Protocol user-plane (GTP-U) packet arriving from the radio access network, which makes its N3 interface both the busiest and the most exposed point of the mobile data path. As operators migrate the core onto Kubernetes and public cloud, the UPF increasingly shares a general-purpose Linux kernel with other workloads, and kernel-bypass frameworks such as DPDK become harder to operate. This paper presents \emph{GTP-Guard}, a design for line-rate GTP-U admission control built on the eXpress Data Path (XDP) hook of the Linux kernel. GTP-Guard drops illegitimate tunnel traffic in the network driver, before socket-buffer allocation, using six ordered stages: peer allow-listing, GTP-U header validation, Tunnel Endpoint Identifier (TEID) admission against session state derived from the Packet Forwarding Control Protocol (PFCP), bounded extension-header parsing, inner-packet anti-spoofing with GTP-in-GTP detection, and per-session policing. We formalize stage ordering as a cost-minimization problem and show that sorting stages by the ratio of per-packet cost to rejection probability minimizes expected per-packet work, which motivates run-time reordering through tail-call program arrays when traffic mix shifts under attack. We further describe a Kubernetes deployment model based on a node-level DaemonSet with pinned eBPF maps that survive agent upgrades, discuss constraints specific to public-cloud virtual NICs, and define a reproducible evaluation methodology, with explicit hypotheses, for throughput, latency, CPU efficiency and attack resilience. This paper presents the design and methodology; experimental results will be reported in a subsequent version.
\end{abstract}


\section{Introduction}

In a 5G Standalone (SA) system, the UPF anchors user sessions and forwards traffic between the radio access network (RAN) and external data networks~\cite{ts23501}. Uplink traffic arrives on the N3 reference point encapsulated in GTP-U over UDP port 2152~\cite{ts29281}, and the UPF is programmed by the Session Management Function (SMF) over the N4 interface using PFCP~\cite{ts29244}. Because the UPF must accept tunnelled packets from a large population of gNBs and, in roaming and multi-UPF topologies, from peer UPFs on N9, its ingress is a natural target for volumetric floods, TEID scanning, spoofed inner sources and GTP-in-GTP injection~\cite{gsmafs20}.

Historically, high-capacity UPFs have relied on kernel-bypass frameworks such as DPDK~\cite{dpdk}, which dedicate polling cores and take the NIC away from the kernel. That model fits poorly with cloud-native operations: Kubernetes schedulers, container network interfaces, observability agents and public-cloud virtual NICs all assume the kernel owns the device. The eXpress Data Path (XDP)~\cite{hoiland2018xdp} offers a middle ground. An XDP program is an eBPF program that runs in the driver's receive path before the kernel allocates a socket buffer, can return a verdict for each frame (drop, pass, transmit, or redirect), and co-exists with the regular networking stack. XDP has been used for in-kernel UPFs~\cite{amaral2021cnsm,amaral2022electronics,eupf} and for GTP-U-aware receive-side scaling~\cite{scheich2023xdp,scheich2024rss}, but considerably less attention has been paid to XDP as a dedicated \emph{admission-control} layer placed in front of an existing UPF, regardless of how that UPF implements forwarding.

This paper addresses that gap. Our contributions are:
\begin{enumerate}
  \item A six-stage XDP pipeline, GTP-Guard, that validates GTP-U traffic against PFCP-derived session state and drops illegitimate packets in the driver (Section~\ref{sec:design}).
  \item A cost model for ordering filter stages, with a closed-form optimal order and a mechanism for re-ordering at run time without detaching the program (Section~\ref{sec:order}).
  \item A Kubernetes deployment model, including state persistence across upgrades, privilege scoping, and public-cloud NIC considerations (Section~\ref{sec:k8s}).
  \item A reproducible evaluation methodology for XDP-based user-plane protection, stated as testable hypotheses (Section~\ref{sec:eval}). Measured results are deferred to a later version of this paper.
\end{enumerate}

\section{Background}

\subsection{eBPF and the XDP Hook}
eBPF lets user space load verified bytecode into the kernel and attach it to hooks. Before a program is accepted, the in-kernel verifier proves memory safety and termination by symbolically executing every path, rejecting out-of-bounds packet access, unbounded loops and unchecked pointer dereferences~\cite{gershuni2019verifier}. Programs share state with user space through \emph{maps}: hash tables, arrays, longest-prefix-match (LPM) tries, per-CPU variants, ring buffers and program arrays used for tail calls~\cite{vieira2020survey}.

XDP is the earliest programmable hook on the receive path~\cite{hoiland2018xdp}. The program receives a context pointing at the raw frame in the driver's DMA buffer and returns one of five actions: \texttt{XDP\_DROP}, \texttt{XDP\_PASS} (continue into the stack), \texttt{XDP\_TX} (bounce out of the same interface), \texttt{XDP\_REDIRECT} (to another interface, CPU, or an AF\_XDP socket), or \texttt{XDP\_ABORTED}. In the original evaluation, an XDP program dropping packets reached roughly 24\,Mpps on a single core~\cite{hoiland2018xdp}, illustrating why dropping early is far cheaper than dropping after socket-buffer allocation.

\subsection{Attach Modes and AF\_XDP}
Table~\ref{tab:modes} summarizes the three attach modes. \emph{Native} mode runs inside the driver and requires driver support; \emph{generic} mode runs after socket-buffer allocation and works on any interface (including veth), but forfeits most of the performance benefit; \emph{offload} mode executes on capable SmartNICs. AF\_XDP sockets~\cite{afxdpdoc} complement XDP by letting a program redirect selected frames into a user-space memory region (UMEM), optionally with zero copy. In GTP-Guard, AF\_XDP is the hand-off path to user-space UPFs, while the filtering itself stays in XDP.

\begin{table}[t]
\centering
\caption{XDP attach modes and their relevance to UPF hosts}
\label{tab:modes}
\footnotesize
\begin{tabular}{@{}>{\raggedright\arraybackslash}p{1.15cm}>{\raggedright\arraybackslash}p{2.9cm}>{\raggedright\arraybackslash}p{3.5cm}@{}}
\toprule
\textbf{Mode} & \textbf{Where it runs} & \textbf{UPF implications} \\
\midrule
Native  & NIC driver, before \texttt{skb} allocation & Target mode for N3; requires driver support and, on some virtual NICs, MTU/queue constraints \\
Generic & After \texttt{skb} allocation & Functional fallback for testing and veth; little performance gain \\
Offload & SmartNIC hardware & Highest efficiency; limited map and helper support \\
\bottomrule
\end{tabular}
\end{table}

\subsection{The 5G User Plane and GTP-U}
The GTPv1-U header~\cite{ts29281} is eight bytes: a flags octet (3-bit version, protocol-type bit, and the E, S and PN bits), a message type, a 16-bit length and a 32-bit TEID. If any of E, S or PN is set, four optional bytes follow (sequence number, N-PDU number and next-extension-header type). On N3, the E bit is normally set because the PDU Session Container extension header carries the QoS Flow Identifier (QFI)~\cite{ts38415}. User data uses message type 255 (G-PDU); signalling messages include Echo Request/Response, Error Indication and End Marker. The UPF learns valid TEIDs, UE addresses and QoS Enforcement Rules (QERs) from PFCP Session Establishment and Modification messages carrying Packet Detection Rules (PDRs) and Forwarding Action Rules (FARs)~\cite{ts29244}.

\section{Related Work}

\textbf{XDP and eBPF data planes.} H{\o}iland-J{\o}rgensen \emph{et al.} introduced XDP and showed that it combines programmability and kernel integration with performance competitive with kernel bypass~\cite{hoiland2018xdp}. Vieira \emph{et al.} survey eBPF/XDP concepts and applications~\cite{vieira2020survey}, and Miano \emph{et al.} report lessons from building chained network services in eBPF~\cite{miano2018complex}. In production, XDP underpins L4 load balancing such as Katran~\cite{katran} and Kubernetes service acceleration in Cilium~\cite{ciliumxdp}.

\textbf{XDP for the 5G user plane.} Navarro do Amaral \emph{et al.} proposed \texttt{upf-bpf}, an in-kernel UPF for MEC environments based on XDP~\cite{amaral2021cnsm}, and later reduced BPF program injection time from about 27\,ms to 1\,ms through a decoupled design while reaching roughly 10--11\,Mpps across six cores~\cite{amaral2022electronics}. Scheich \emph{et al.} extended a UPF with XDP GTP-U handling and an XDP receive-side-scaling scheme keyed on GTP-U headers~\cite{scheich2023xdp,scheich2024rss}. The open-source eUPF attaches an XDP program that performs PDR detection, gating, QoS and forwarding and has been tested with free5GC, Open5GS and OpenAirInterface~\cite{eupf}. Deshpande and Bera analysed eBPF-enabled containerized 5G cores and showed that the shared-kernel model exposes network functions to cross-container threats~\cite{deshpande2026vuln}.

\textbf{Positioning.} Prior work uses XDP to \emph{be} the UPF. GTP-Guard instead treats XDP as a UPF-agnostic \emph{front door}: it protects any UPF, kernel-based, AF\_XDP-based or DPDK-adjacent, by admitting only traffic consistent with current PFCP state, and it focuses on stage ordering, run-time adaptation and cloud-native lifecycle concerns not addressed in earlier designs.

\section{Threat Model}

We consider an adversary who can inject packets towards the UPF's N3 or N9 address, either from a compromised or misconfigured RAN node, from a subscriber device that crafts packets inside its own tunnel, or from any network segment that can reach the transport network. We assume the N4 control channel and the Kubernetes control plane are trusted. Table~\ref{tab:threats} maps the threats we address to pipeline stages.

\begin{table}[t]
\centering
\caption{Threats addressed and the stage that mitigates each}
\label{tab:threats}
\footnotesize
\begin{tabular}{@{}>{\raggedright\arraybackslash}p{2.6cm}p{0.6cm}p{1.6cm}p{2.0cm}@{}}
\toprule
\textbf{Threat} & \textbf{Stage} & \textbf{Map} & \textbf{Verdict} \\
\midrule
Off-path source flooding N3 & S1 & LPM trie & \texttt{DROP} \\
Malformed / fuzzed GTP-U & S2, S4 & -- & \texttt{DROP} \\
TEID scanning / stale tunnels & S3 & Hash (TEID) & \texttt{DROP} \\
Inner source spoofing & S5 & Hash (TEID) & \texttt{DROP} \\
GTP-in-GTP injection & S5 & -- & \texttt{DROP} \\
Per-subscriber flood & S6 & Per-CPU array & \texttt{DROP} \\
\bottomrule
\end{tabular}
\end{table}

Out of scope are attacks on PFCP itself, encrypted payload attacks, and denial of service that saturates the physical link before the host; the latter must be absorbed upstream.

\section{GTP-Guard Design}
\label{sec:design}

\subsection{Design Goals}
GTP-Guard is guided by four goals: (G1) \emph{drop early}, so that every rejected packet costs as little as possible; (G2) \emph{stay UPF-agnostic}, so the guard can front any UPF implementation; (G3) \emph{never be the source of truth}, so that session state is always derived from PFCP and a guard failure can degrade to pass-through; and (G4) \emph{be operable on Kubernetes}, with declarative configuration, rolling upgrades and Prometheus-native telemetry.

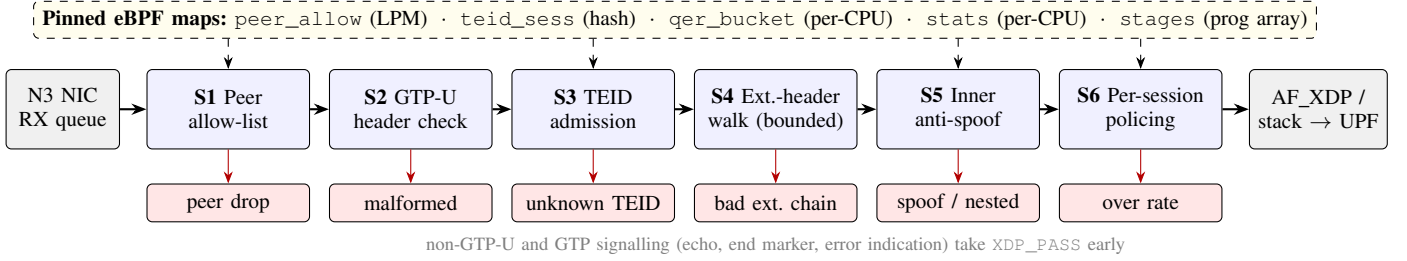
\begin{figure*}[t]
\centering
\begin{tikzpicture}[
  font=\footnotesize,
  stage/.style={draw, rounded corners=2pt, fill=blue!6, minimum width=2.15cm, minimum height=1.05cm, align=center},
  io/.style={draw, rounded corners=2pt, fill=gray!12, minimum width=1.5cm, minimum height=1.05cm, align=center},
  drop/.style={draw, rounded corners=2pt, fill=red!10, minimum width=2.15cm, minimum height=0.5cm, align=center},
  arr/.style={-{Stealth[length=2mm]}, thick},
  darr/.style={-{Stealth[length=1.6mm]}, red!70!black}
]
\node[io] (nic) {N3 NIC\\RX queue};
\node[stage, right=0.35cm of nic] (s1) {\textbf{S1} Peer\\allow-list};
\node[stage, right=0.25cm of s1] (s2) {\textbf{S2} GTP-U\\header check};
\node[stage, right=0.25cm of s2] (s3) {\textbf{S3} TEID\\admission};
\node[stage, right=0.25cm of s3] (s4) {\textbf{S4} Ext.-header\\walk (bounded)};
\node[stage, right=0.25cm of s4] (s5) {\textbf{S5} Inner\\anti-spoof};
\node[stage, right=0.25cm of s5] (s6) {\textbf{S6} Per-session\\policing};
\node[io, right=0.35cm of s6] (out) {AF\_XDP /\\stack $\rightarrow$ UPF};

\foreach \a/\b in {nic/s1,s1/s2,s2/s3,s3/s4,s4/s5,s5/s6,s6/out} \draw[arr] (\a) -- (\b);

\node[drop, below=0.45cm of s1] (d1) {peer drop};
\node[drop, below=0.45cm of s2] (d2) {malformed};
\node[drop, below=0.45cm of s3] (d3) {unknown TEID};
\node[drop, below=0.45cm of s4] (d4) {bad ext.\ chain};
\node[drop, below=0.45cm of s5] (d5) {spoof / nested};
\node[drop, below=0.45cm of s6] (d6) {over rate};
\foreach \s/\d in {s1/d1,s2/d2,s3/d3,s4/d4,s5/d5,s6/d6} \draw[darr] (\s) -- (\d);

\node[draw, dashed, rounded corners=2pt, fill=yellow!8, above=0.4cm of s3.north east, anchor=south, minimum width=9.5cm, align=center] (maps)
  {\textbf{Pinned eBPF maps:} \texttt{peer\_allow} (LPM) \;$\cdot$\; \texttt{teid\_sess} (hash) \;$\cdot$\; \texttt{qer\_bucket} (per-CPU) \;$\cdot$\; \texttt{stats} (per-CPU) \;$\cdot$\; \texttt{stages} (prog array)};
\draw[dashed, -{Stealth[length=1.6mm]}] (maps.south -| s1) -- (s1.north);
\draw[dashed, -{Stealth[length=1.6mm]}] (maps.south -| s3) -- (s3.north);
\draw[dashed, -{Stealth[length=1.6mm]}] (maps.south -| s5) -- (s5.north);
\draw[dashed, -{Stealth[length=1.6mm]}] (maps.south -| s6) -- (s6.north);

\node[font=\scriptsize, gray, below=0.08cm of d4] {non-GTP-U and GTP signalling (echo, end marker, error indication) take \texttt{XDP\_PASS} early};
\end{tikzpicture}
\caption{GTP-Guard XDP pipeline on the N3 receive path. Each stage either rejects in the driver or hands the frame to the next stage; only admitted G-PDUs reach the UPF.}
\label{fig:pipeline}
\end{figure*}

\subsection{Pipeline}
Fig.~\ref{fig:pipeline} shows the pipeline, and Fig.~\ref{lst:xdp} gives a condensed implementation of the IPv4 path.

\textbf{Pre-classification.} Frames that are not IPv4/UDP to port 2152 receive \texttt{XDP\_PASS} immediately, so the guard is transparent to management, PFCP and N6 traffic sharing the interface.

\textbf{S1 -- Peer allow-list.} The outer source address is looked up in an LPM trie holding gNB and N9 peer prefixes. The trie is populated from operator inventory (for example, via a Kubernetes custom resource) rather than from PFCP, because gNB addressing changes far less frequently than sessions.

\textbf{S2 -- Header validation.} The guard enforces version~1 and the protocol-type bit, and passes non-G-PDU messages (Echo, End Marker, Error Indication) to the UPF unmodified, since their volume is low and their semantics belong to the UPF.

\textbf{S3 -- TEID admission.} The TEID indexes a hash map whose entries mirror the uplink PDRs currently installed by the SMF. A miss means the tunnel does not exist on this UPF and the packet is dropped. This single lookup removes TEID scanning and traffic for torn-down sessions.

\textbf{S4 -- Extension-header walk.} The optional fields and extension-header chain are traversed in a loop bounded by a compile-time constant, satisfying the verifier while accommodating the PDU Session Container that is ubiquitous on N3~\cite{ts38415}. Chains that exceed the bound or overrun the frame are rejected.

\textbf{S5 -- Inner checks.} The inner source address must equal the UE address bound to the TEID, which defeats spoofing of other subscribers' addresses. An inner UDP datagram addressed to port 2152 indicates GTP-in-GTP injection and is dropped.

\textbf{S6 -- Policing.} A token bucket keyed by QER index enforces a per-session ceiling. Admitted packets are redirected to an AF\_XDP socket bound to the receive queue when the UPF is an AF\_XDP application; otherwise the \texttt{XDP\_PASS} fallback of \texttt{bpf\_redirect\_map} delivers them to the kernel stack.

\begin{figure*}[t]
\begin{lstlisting}[style=ebpf]
SEC("xdp")
int gtp_guard(struct xdp_md *ctx)
{
    void *data = (void *)(long)ctx->data, *end = (void *)(long)ctx->data_end;
    struct ethhdr *eth = data;
    if ((void *)(eth + 1) > end)                     return drop(R_MALFORMED);
    if (eth->h_proto != bpf_htons(ETH_P_IP))         return XDP_PASS;
    struct iphdr *oip = (void *)(eth + 1);
    if ((void *)(oip + 1) > end)                     return drop(R_MALFORMED);
    if (oip->protocol != IPPROTO_UDP || oip->ihl != 5) return XDP_PASS;
    struct udphdr *udp = (void *)(oip + 1);
    if ((void *)(udp + 1) > end)                     return drop(R_MALFORMED);
    if (udp->dest != bpf_htons(GTPU_PORT))           return XDP_PASS;      /* not GTP-U */

    struct lpm_v4 pk = { .prefixlen = 32, .addr = oip->saddr };          /* S1: peer allow-list */
    if (!bpf_map_lookup_elem(&peer_allow, &pk))      return drop(R_PEER);

    struct gtpu_hdr *g = (void *)(udp + 1);                               /* S2: header sanity */
    if ((void *)(g + 1) > end)                       return drop(R_MALFORMED);
    if ((g->flags >> 5) != 1 || !(g->flags & GTP_PT)) return drop(R_HDR);
    if (g->type != GTP_GPDU)                         return XDP_PASS;      /* echo, end marker, err-ind */

    struct sess *s = bpf_map_lookup_elem(&teid_sess, &g->teid);           /* S3: TEID admission */
    if (!s)                                          return drop(R_TEID);

    __u8 *p = (__u8 *)(g + 1);                                            /* S4: bounded ext-header walk */
    if (g->flags & (GTP_E | GTP_S | GTP_PN)) {
        if (p + 4 > (__u8 *)end)                     return drop(R_MALFORMED);
        __u8 nxt = (g->flags & GTP_E) ? p[3] : 0;
        p += 4;
        #pragma unroll
        for (int i = 0; i < MAX_EXT_HDRS; i++) {
            if (!nxt) break;
            if (p + 4 > (__u8 *)end)                 return drop(R_MALFORMED);
            __u32 len = (__u32)p[0] * 4;              /* length in 4-octet units */
            if (len == 0 || len > MAX_EXT_LEN || p + len > (__u8 *)end) return drop(R_MALFORMED);
            nxt = p[len - 1];
            p += len;
        }
        if (nxt)                                     return drop(R_HDR);   /* chain longer than bound */
    }

    struct iphdr *iip = (void *)p;                                        /* S5: inner checks */
    if ((void *)(iip + 1) > end)                     return drop(R_MALFORMED);
    if (iip->version != 4)                           return XDP_PASS;      /* IPv6 path omitted */
    if (iip->saddr != s->ue_ipv4)                    return drop(R_SPOOF);
    if (iip->protocol == IPPROTO_UDP && iip->ihl == 5) {
        struct udphdr *iu = (void *)(iip + 1);
        if ((void *)(iu + 1) <= end && iu->dest == bpf_htons(GTPU_PORT)) return drop(R_NESTED);
    }

    if (!police(s->qer_idx, ctx->data_end - ctx->data)) return drop(R_RATE); /* S6: policing */
    return bpf_redirect_map(&xsks, ctx->rx_queue_index, XDP_PASS);          /* AF_XDP, else stack */
}
\end{lstlisting}
\caption{Condensed GTP-Guard XDP program (IPv4 path). \texttt{drop()} increments a per-CPU reason counter and returns \texttt{XDP\_DROP}; map definitions and the IPv6 path are omitted for space.}
\label{lst:xdp}
\end{figure*}

\subsection{State Model and Control-Plane Synchronization}
GTP-Guard holds three kinds of state: the peer allow-list (operator-managed), session state (PFCP-derived) and policing state (data-plane-owned). Session entries map a TEID to a compact record containing the UE IPv4 address, an IPv6 prefix, and a QER index. To satisfy G3, a user-space agent subscribes to session lifecycle events exported by the UPF, translating PFCP Session Establishment, Modification and Deletion into map updates. Where the UPF exposes no such interface, the agent can derive the same state from the N4 exchange it terminates; we prefer the event-export path because it avoids re-implementing PFCP parsing and remains correct if N4 is protected by IPsec.

Two consistency properties matter. First, \emph{insert-before-use}: a TEID must be present before the gNB sends traffic on it. Because PFCP establishment precedes the N2 signalling that hands the TEID to the gNB, an agent that applies updates within the PFCP round trip satisfies this naturally. Second, \emph{delete-after-drain}: entries are removed only after session deletion is acknowledged, so in-flight uplink packets are not penalized. A watchdog compares map cardinality with the UPF's session count and flips the guard into a pass-through mode through a single control-map flag if drift exceeds a threshold, preferring availability over filtering when state cannot be trusted.

\subsection{Stage Ordering}
\label{sec:order}
Stages are independent tests that each either reject or forward a packet. Let stage $i$ cost $c_i$ cycles and let $p_i$ be the probability that a packet reaching it is forwarded. For an order $\pi$, the expected per-packet cost is
\begin{equation}
  E[C_\pi] \;=\; \sum_{k=1}^{n} c_{\pi(k)} \prod_{j=1}^{k-1} p_{\pi(j)} .
  \label{eq:cost}
\end{equation}

\begin{proposition}
$E[C_\pi]$ is minimized by ordering stages in ascending order of $r_i = c_i/(1-p_i)$, with $r_i=\infty$ when $p_i=1$.
\end{proposition}
\begin{proof}
Consider two adjacent stages $a$ followed by $b$, preceded by stages with joint forwarding probability $P$. Their contribution is $P(c_a + p_a c_b)$, and swapping gives $P(c_b + p_b c_a)$. Placing $a$ first is no worse iff $c_a + p_a c_b \le c_b + p_b c_a$, i.e.\ $c_a(1-p_b) \le c_b(1-p_a)$, i.e.\ $r_a \le r_b$. Stages outside the pair are unaffected, so any order containing an inversion can be improved by an adjacent swap, and the sorted order is optimal.
\end{proof}

Under benign load almost every packet passes every stage ($p_i \approx 1$), so ordering barely matters. Under attack it matters a great deal: in a TEID-scanning flood, S3 rejects most traffic while S1 rejects almost none, so $r_3 \ll r_1$ even though the LPM lookup is cheaper. Dependencies constrain the order (S3--S6 need the parsed GTP-U header), but S1, S3 and S6 can be permuted. GTP-Guard implements each permutable stage as a separate program in a \texttt{BPF\_MAP\_TYPE\_PROG\_ARRAY} and chains them with tail calls. The agent estimates $p_i$ from per-stage counters, recomputes $r_i$, and atomically rewrites program-array slots, re-ordering the pipeline without detaching XDP from the interface. The tail-call indirection adds a small fixed cost per stage, so the agent uses a single fused program in steady state and switches to the re-orderable chain only when drop rates cross a configured threshold.

\subsection{Policing and Receive-Side Scaling}
Per-session policing must reconcile two facts. Token buckets require shared state, but XDP runs concurrently on every receive queue, and atomic updates to a shared bucket serialize cores. Conversely, NIC receive-side scaling usually hashes the outer five-tuple, which for GTP-U carries little entropy when few gNBs serve many subscribers, so one subscriber's packets generally land on one queue. GTP-Guard therefore keeps buckets in a per-CPU array and assigns each CPU a share of the session rate. When outer-header entropy is too low to spread load, a dispatcher stage redirects packets to a CPU map keyed by a hash of the TEID, the technique shown by Scheich \emph{et al.} to improve uplink capacity substantially~\cite{scheich2023xdp,scheich2024rss}; this also makes per-CPU bucket shares exact, since all packets of a session then reach the same CPU.

\section{Cloud-Native Deployment on Kubernetes}
\label{sec:k8s}

\begin{figure}[t]
\centering
\begin{tikzpicture}[
  font=\scriptsize,
  box/.style={draw, rounded corners=2pt, align=center, minimum height=0.75cm, minimum width=2.1cm},
  arr/.style={-{Stealth[length=1.6mm]}},
  lbl/.style={font=\tiny, fill=white, inner sep=1pt}
]
\node[box, fill=gray!12, minimum width=1.0cm] (gnb) at (0,2.5) {gNB};
\node[box, fill=orange!12] (nic) at (2.9,2.5) {N3 NIC (VF/ENI)\\+ XDP GTP-Guard};
\node[box, fill=green!10]  (upf) at (6.4,2.5) {UPF pod\\(AF\_XDP or kernel)};
\node[box, fill=blue!8]    (maps) at (2.9,1.0) {Pinned maps\\\texttt{/sys/fs/bpf}};
\node[box, fill=blue!12]   (agent) at (6.4,1.0) {Guard agent\\(DaemonSet)};
\node[box, fill=gray!12]   (smf) at (6.4,4.1) {SMF};
\node[box, fill=gray!12]   (crd) at (2.9,-0.75) {Policy CR\\(GitOps)};
\node[box, fill=gray!12]   (prom) at (6.4,-0.75) {Prometheus};

\begin{scope}[on background layer]
\node[draw, dashed, rounded corners=3pt, fit=(nic)(maps)(upf)(agent), inner sep=6pt,
      label={[font=\scriptsize]above left:Kubernetes worker node}] (node) {};
\end{scope}

\draw[arr, thick] (gnb) -- node[lbl, above=1pt]{N3} (nic);
\draw[arr, thick] (nic) -- node[lbl, above=1pt]{admitted} (upf);
\draw[arr] (smf) -- node[lbl, right=2pt]{N4 PFCP} (upf);
\draw[arr] (upf) -- node[lbl, right=2pt]{session events} (agent);
\draw[arr] (agent) -- node[lbl, above=1pt]{map updates} (maps);
\draw[arr] (maps) -- node[lbl, left=2pt]{lookups} (nic);
\draw[arr] (crd) -- node[lbl, above=1pt, sloped]{allow-list, limits} (agent);
\draw[arr] (agent) -- node[lbl, right=2pt]{metrics} (prom);
\end{tikzpicture}
\caption{Deployment on a Kubernetes worker node. The agent owns program lifecycle and map population; the XDP program and its maps outlive agent restarts.}
\label{fig:k8s}
\end{figure}
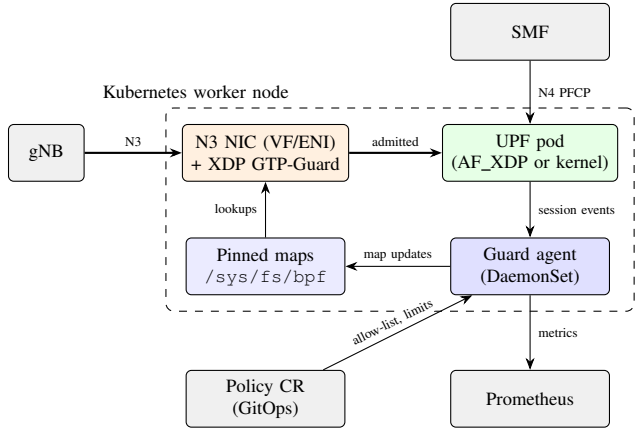

\subsection{Packaging and Lifecycle}
GTP-Guard ships as a DaemonSet scheduled onto UPF nodes by node affinity (Fig.~\ref{fig:k8s}). The agent loads a CO-RE (compile once, run everywhere) object built against BTF, so a single image runs across kernel versions in a mixed fleet. Maps and links are pinned under \texttt{/sys/fs/bpf}; on restart, the agent re-opens pinned objects rather than re-creating them, so a rolling upgrade of the agent never interrupts filtering or loses session state, the same persistence property eUPF relies on for datapath continuity~\cite{eupf}. Program upgrades use atomic replacement of the attached link, so the interface is never left without a program. Where other XDP programs must share the interface, the libxdp multi-program dispatcher~\cite{libxdp} chains them with defined priorities.

\subsection{Interface Attachment}
On bare-metal or private-cloud clusters, N3 is typically an SR-IOV virtual function plumbed into the UPF pod through Multus. Native XDP must then attach to the VF in the host namespace before it is moved, or inside the pod namespace by an init container; the latter couples guard lifecycle to the UPF and is therefore avoided. On public cloud, N3 is usually a dedicated elastic network interface on the node, and the guard attaches in the host namespace.

\subsection{Privilege Scoping}
Loading XDP requires \texttt{CAP\_BPF} and \texttt{CAP\_NET\_ADMIN} on modern kernels; older kernels require \texttt{CAP\_SYS\_ADMIN}. Deshpande and Bera note that eBPF-enabled 5G pods commonly run with broad capabilities and that the shared kernel then becomes an attack path~\cite{deshpande2026vuln}. GTP-Guard confines all privileged operations to the agent, keeps the UPF container unprivileged when it uses the kernel path, restricts map access by pinning under a directory readable only by the agent, and exposes telemetry through an unprivileged sidecar.

\subsection{Public-Cloud NIC Considerations}
Virtual NICs impose constraints absent on bare metal. Native XDP support in cloud NIC drivers may require a reduced MTU, a limited number of queues, or specific buffer sizes, and loaders that do not meet these requirements can silently fall back to generic mode. A practitioner report on AWS describes AF\_XDP on the ENA driver falling back to generic mode with about a quarter of packets lost unnoticed, and reaching lossless operation only after adjusting buffer sizing, at which point the instance's packets-per-second allowance became the bottleneck~\cite{toonk2026afxdp}. GTP-Guard's agent therefore (i) requests native mode explicitly and fails closed to an alert rather than silently accepting generic mode, (ii) exports the effective attach mode and per-queue counters as metrics, and (iii) treats the provider's per-instance packet-rate allowance as a first-class capacity limit in planning.

\section{Evaluation Methodology}
\label{sec:eval}

\subsection{Per-Packet Budget}
For minimum-size Ethernet frames, line rate is 14.88\,Mpps at 10\,Gb/s, 37.2\,Mpps at 25\,Gb/s and 148.8\,Mpps at 100\,Gb/s. At a 3\,GHz clock, one core sustaining 14.88\,Mpps has a budget of about 200 cycles per packet. Real N3 frames are larger, because GTP-U, the PDU Session Container and inner headers add at least 50 bytes, so the packet rate per link is lower; however, attack traffic need not respect realistic sizes. We therefore report both packet rate and bit rate, and treat cycles per packet as the primary efficiency metric.

\subsection{Testbed and Workloads}
The reference testbed comprises a traffic generator capable of stateless GTP-U generation (e.g., TRex with GTP-U templates), a device under test running a recent long-term-support kernel with a native-XDP-capable NIC, an open-source 5G core (Open5GS or free5GC) with a UPF in either kernel or AF\_XDP mode, and UERANSIM or PacketRusher for session establishment. A second profile repeats the experiments on a managed Kubernetes cluster in public cloud. Workloads are:
\begin{itemize}
  \item \textbf{W1 Benign:} $10^4$--$10^6$ active TEIDs, IMIX-like inner payloads.
  \item \textbf{W2 TEID scan:} benign background plus random-TEID G-PDUs at rates up to line rate.
  \item \textbf{W3 Spoof:} valid TEIDs carrying forged inner sources.
  \item \textbf{W4 Nested:} GTP-in-GTP inner datagrams.
  \item \textbf{W5 Churn:} PFCP establishment and deletion at increasing rates, measuring map-update latency and insert-before-use violations.
\end{itemize}

\subsection{Baselines and Metrics}
We compare (B0) no filtering, (B1) equivalent rules in \texttt{nftables}, (B2) GTP-Guard in generic mode, (B3) GTP-Guard in native mode with a fused program, and (B4) native mode with adaptive ordering. Metrics are zero-loss throughput (RFC~2544 style), drop capacity per core, cycles per packet (from \texttt{perf} and \texttt{bpftool prog profile}), p50/p99/p99.9 added latency for admitted packets, CPU consumed by the UPF under attack, and false-drop rate for legitimate traffic. Table~\ref{tab:hyp} lists the hypotheses each experiment is designed to test. Literature figures in Table~\ref{tab:lit} provide context for expected magnitudes but are not directly comparable because hardware and pipelines differ.

\begin{table}[t]
\centering
\caption{Reported XDP performance in prior work (context only)}
\label{tab:lit}
\footnotesize
\begin{tabular}{@{}>{\raggedright\arraybackslash}p{2.0cm}>{\raggedright\arraybackslash}p{5.8cm}@{}}
\toprule
\textbf{Work} & \textbf{Reported result} \\
\midrule
XDP~\cite{hoiland2018xdp} & $\approx$24\,Mpps drop on a single core \\
upf-bpf~\cite{amaral2022electronics} & $\approx$10\,Mpps downlink, 11\,Mpps uplink on six cores; program injection reduced from 27\,ms to 1\,ms \\
XDP GTP-U ext.~\cite{scheich2023xdp} & XDP UPF throughput well above user-space UPF; TEID-keyed CPU redirect improves uplink by about $2.5\times$ \\
\bottomrule
\end{tabular}
\end{table}

\begin{table}[t]
\centering
\caption{Hypotheses the evaluation is designed to test}
\label{tab:hyp}
\footnotesize
\begin{tabular}{@{}p{0.45cm}>{\raggedright\arraybackslash}p{5.0cm}>{\raggedright\arraybackslash}p{1.9cm}@{}}
\toprule
\textbf{ID} & \textbf{Hypothesis} & \textbf{Test} \\
\midrule
H1 & Native-mode GTP-Guard adds negligible latency and no throughput loss for benign traffic relative to B0 & W1: B0 vs.\ B3 \\
H2 & Native-mode drop capacity per core substantially exceeds \texttt{nftables} and generic mode & W2: B1, B2, B3 \\
H3 & Early dropping keeps UPF CPU near its benign level during TEID-scan floods & W2: B0 vs.\ B3 \\
H4 & Adaptive ordering lowers cycles per packet versus the fused program when the attack mix shifts & W2: B3 vs.\ B4 \\
H5 & PFCP-to-map latency stays below the N2 signalling delay, so no insert-before-use violations occur & W5 \\
\bottomrule
\end{tabular}
\end{table}

\section{Discussion and Limitations}

\textbf{Verifier constraints.} The 512-byte stack, bounded loops and instruction complexity limits shape the design: extension-header walks must be bounded, and IPv6 inner packets with extension headers need a separately bounded parser. Splitting stages into tail-called programs also keeps each program's verification cost manageable.

\textbf{Fragmentation.} Outer IP fragments carry no UDP header after the first fragment and cannot be validated in XDP. GTP-Guard passes fragments to the stack, where reassembly occurs before the UPF; the transport MTU should be engineered so that N3 fragmentation does not occur in normal operation, and fragment rate is exported as a metric because a surge is itself an attack signal.

\textbf{State divergence.} A guard that drops on TEID miss can black-hole legitimate traffic if session state lags PFCP. The watchdog and pass-through flag bound the impact, but operators should monitor insert-to-first-packet latency (W5) before enabling strict mode.

\textbf{Shared-kernel risk.} XDP improves isolation of the UPF from hostile traffic, but it enlarges the trusted computing base on the node. Loading must remain limited to a single audited agent, and program signing should be adopted where the platform supports it~\cite{deshpande2026vuln}.

\textbf{Scope of filtering.} GTP-Guard does not inspect payloads, enforce full PDR precedence, or apply downlink policy on N6; the latter can reuse the same framework with a UE-address-keyed map.

\section{Conclusion and Future Work}

GTP-Guard places a PFCP-aware, line-rate admission layer in the driver receive path of a Kubernetes-hosted UPF, removing illegitimate GTP-U traffic before it consumes socket buffers or UPF cycles while remaining agnostic to the UPF's own forwarding technology. The stage-ordering result gives a principled basis for reconfiguring the pipeline when traffic shifts under attack, and the deployment model addresses lifecycle, privilege and public-cloud NIC constraints that determine whether XDP is practical in production. The immediate next step is to execute the evaluation of Section~\ref{sec:eval} and test hypotheses H1--H5; further work will extend the pipeline to IPv6 and N6, evaluate hardware offload on SmartNICs, and close the loop between per-stage telemetry and an automated controller that adjusts allow-lists, ordering and policing thresholds as part of a self-healing operations framework.

\bibliographystyle{IEEEtran}
\bibliography{references}

\end{document}